\documentclass[conference,letterpaper]{IEEEtran}

\usepackage{color}
\usepackage{enumerate}
\usepackage{booktabs}
\usepackage{cite}
\usepackage{setspace}
\usepackage{bm}
\usepackage{amsthm}
\usepackage[cmex10]{amsmath}
\usepackage{colortbl}
\usepackage{stfloats}
\usepackage{subfigure}
\usepackage{amssymb}
\usepackage{algorithm}
\usepackage{algorithmic}
\usepackage{multirow}
\usepackage{xcolor}
\usepackage{graphicx}
\usepackage{geometry}
\usepackage{stmaryrd}
\usepackage{caption}
\usepackage{subfigure}
\theoremstyle{definition}

\newtheorem{lemma}{Lemma}

\newtheorem{theorem}{Theorem}

\usepackage{epstopdf}
\IEEEoverridecommandlockouts

\begin{document}
%
\title{Beamforming Towards Seamless Sensing Coverage for  Cellular Integrated Sensing and Communication}

\author{ \IEEEauthorblockN{Ruoguang Li\IEEEauthorrefmark{1}, Zhiqiang Xiao\IEEEauthorrefmark{1}, and Yong Zeng\IEEEauthorrefmark{1}\IEEEauthorrefmark{2}\\}
\IEEEauthorblockN{\IEEEauthorrefmark{1}National Mobile Communications Research Laboratory, Southeast University, Nanjing, P. R. China\\}
\IEEEauthorblockN{\IEEEauthorrefmark{2}Purple Mountain Laboratories, Nanjing, P. R. China\\}
\IEEEauthorblockN{\{ruoguangli, zhiqiang\_xiao, yong\_zeng\}@seu.edu.cn\\}
}
\maketitle

\begin{abstract}
The sixth generation (6G) mobile communication networks are expected to offer a new paradigm of cellular integrated sensing and communication (ISAC). However, due to the intrinsic difference between sensing and communication in terms of coverage requirement, current cellular networks that are deliberately planned mainly for communication coverage are difficult to achieve seamless sensing coverage. To address this issue, this paper studies the beamforming optimization towards seamless sensing coverage for a basic bi-static ISAC system, while ensuring that the communication requirements of multiple users equipment (UEs) are satisfied. Towards this end, an optimization problem is formulated to maximize the worst-case sensing signal-to-noise ratio (SNR) in a prescribed coverage region, subject to the signal-to-interference-plus-noise ratio (SINR) requirement for each UE. To gain some insights, we first investigate the special case with one single UE and one single sensing point, for which a closed-from expression of the optimal beamforming is obtained. For the general case with  multiple communication UEs and contiguous regional sensing coverage, an efficient algorithm based on successive convex approximation (SCA) is proposed to solve the non-convex beamforming optimization problem. Numerical results demonstrate that the proposed design is able to achieve seamless sensing coverage in the prescribed region, while guaranteeing the communication requirements of the UEs.
\end{abstract}


%

\section{Introduction}

Integrated sensing and communication (ISAC) is envisioned to become an important feature for the sixth generation (6G) mobile communication networks \cite{Fliu2021}, in which radar sensing and communication capabilities are fully integrated with shared platform and radio resources. In particular, enabling ISAC with cellular infrastructure is a promising 
framework to achieve enhanced sensing and communication performance via cooperation among base stations (BSs) \cite{JAZhang2022}, which has a wide range of potential use cases such as high-accuracy localization, autonomous driving, beamforming alignment, etc. Therefore, extensive research efforts have been devoted to the study of cellular ISAC from the perspectives of information theory, signal processing, waveform design, and resource management, etc \cite{ALiu2021, JAZhang2021, ZXiao2022, NCLuong2021}.

However, the coverage performance, which is a prerequisite to offer ``always-available'' ISAC service in the prescribed region, has received little attention. Compared to the standalone communication or sensing networks, the coverage issue for ISAC is more complicated due to the intrinsic difference between communication and sensing in terms of coverage requirement. Specifically, in contrast to the typically one-hop signal propagation for wireless communication, the signal attenuation for radar sensing generally involves two hops by target reflection, indicating heterogeneous coverage range between such two functionalities. Secondly, the coverage of wireless communication is usually BS-centric with disk-like shape. In contrast, the coverage shapes of radar sensing critically depend on the radar deployment mode. 
For example, in the mono-static radar mode, the coverage shape is similar to that in the wireless communication, but the actual coverage radius depends on the radar cross-section (RCS) for the targets of interest \cite{EFKnott1993}. Besides, for bi-static radar mode without dedicated coverage optimization, the coverage is typically characterized by a series of signal-to-noise ratio (SNR) contours with diverse shapes known as \emph{Cassini oval} \cite{MLSkolnik1990}, as illustrated in Fig. \ref{Fig_cassini_oval}. The sensing coverage issue is even more complicated for multi-static radar mode since the sensing regions of different radars are coupled with each other \cite{SRDoughty2008}.  Therefore, current cellular networks that are mainly planned for communication coverage are difficult to  achieve seamless sensing coverage for future ISAC systems.
\vspace{-4mm}
\begin{figure}[h]
\center
\includegraphics[width=2.2in]{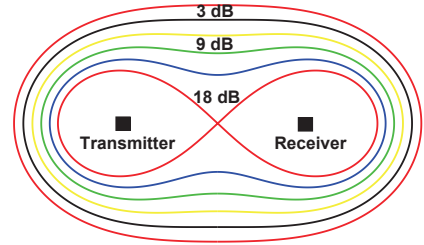}\\
\caption{Cassini oval for SNR contours of  bi-static radar without sensing coverage optimization \cite{MLSkolnik1990}.}\label{Fig_cassini_oval}
\end{figure}
\vspace{-3mm}

To address the above issue, in this paper, we study a basic bi-static cellular ISAC system to provide seamless sensing coverage for the prescribed region, while satisfying the communication requirements of multiple users equipment (UEs). To that end, a beamforming optimization problem is formulated to maximize the worst-case sensing SNR in the prescribed region, subject to the signal-to-interference-plus-noise ratio (SINR) requirements for each communication UE. The formulated problem is non-convex and difficult to be efficiently solved in general.  To gain some insights, we first investigate the special case with one single UE and one single sensing point, for which a closed-from expression of the optimal beamforming is obtained. For the general case with multiple UEs and contiguous regional sensing coverage, an efficient algorithm based on successive convex approximation (SCA) technique is proposed for the non-convex optimization problem. Numerical results demonstrate that the proposed design is able to achieve seamless sensing coverage in the prescribed region, while guaranteeing the communication requirements of the UEs.

\section{System Model and Problem Formulation}
\vspace{-4mm}
\begin{figure}[h]
\center
\includegraphics[width=2.7in]{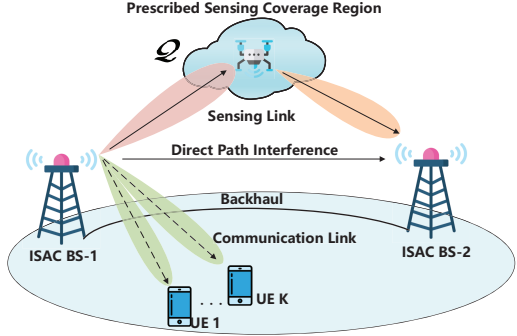}\\
\caption{Bi-static ISAC with prescribed coverage region.}\label{Fig_1}
\end{figure}
\vspace{-3mm}
As illustrated in Fig. \ref{Fig_1}, we consider a basic bi-static ISAC system, which consists of two multi-antenna BSs that are inter-connected via backhaul links.  We assume that BS-1 corresponds to the ISAC transmitter, while BS-2 is the sensing receiver that receives the reflected signals from the sensing target. Let $M_t$ and $M_r$ denote the number of antennas at BS-1 and BS-2, respectively. The objective is to communicate with $K$ single-antenna UEs, denoted by the set $\mathcal{K}=\{1,...,K\}$,  where $K\leq M_t$, while ensuring the seamless sensing coverage for the prescribed region, denoted by $\mathcal{Q}$. Without loss of generality, we consider a three-dimensional (3D) Cartesian coordinate system, in which BS-1 and BS-2 are located at $\mathbf{o}=(0,0,H)$ and $\mathbf{o}'=(0,D,
H)$, respectively, where $H$ is the BS height and $D$ is the inter-BS distance. Furthermore, let ${\mathbf q}\in \mathcal{Q}$ denote an arbitrary location within the prescribed coverage region $\mathcal{Q}$. 

Let $\mathbf{s}(t)=[s_1(t),...,s_{M_t}(t)]^T\in \mathbb{C}^{M_t\times 1}$ denote the signal vector at BS-1, in which $[s_1(t),...,s_{K}(t)]^T\in \mathbb{C}^{K\times 1}$ is the zero-mean unit-power information-bearing signal vector associated with the $K$ communication UEs, and $[s_{K+1}(t),...,s_{M_t}(t)]^T\in \mathbb{C}^{(M_t-K)\times 1}$ is the dedicated radar waveform vector.
Note that since the BSs are inter-connected via backhaul links, the information-bearing signals $[s_{1}(t),...,s_{K}(t)]^T$ can be made available at the sensing receiver BS-2. This makes it possible to utilize the communication signals for sensing coverage as well. Furthermore, to fully exploit the spatial degree of freedom (DoF) of the $M_t$ transmit antennas, without loss of generality, dedicated sensing waveforms $[s_{K+1}(t),...,s_{M_t}(t)]^T$ are also included in the problem formulation. 

Since $s_k(t), k\in \{1,...,K\}$ are random communication signals that are independent for different UEs, we have
\begin{align}
\small
\label{comm_auto}
    {\rm{\bf E}}\left[s_k(t)s^{*}_{\tilde k}(t-\tau)\right]=
\left\{
    \begin{array}{lr}
        R_{C}(\tau), &k=\tilde k,\\
        0, &k\neq\tilde k,\\
     \end{array}
\right.
\end{align}
$k,{\tilde k}\in \{1,...,K\}$, where $R_{C}(\tau)$ is the autocorrelation function of the random communication signals $s_k(t), k\in \{1,...,K\}$ with normalized power $R_{C}(0)=1$. On the other hand,  the dedicated radar waveforms $s_k(t),k\in \{K+1,...,M_t\}$ are deterministic, which are orthogonal with each other, namely,
\begin{align}
\small
\label{radar_auto}
    \frac{1}{T_p}\int_{T_p}s_k(t)s^{*}_{\tilde k}(t-\tau)dt=
\left\{
    \begin{array}{lr}
        R_S(\tau), &k=\tilde k,\\
        0, &k\neq\tilde k.\\
     \end{array}
\right.
\end{align}
$k,{\tilde k}\in \{K+1,...,M_t\}$, where $R_{S}(\tau)$ denotes autocorrelation function with normalized power $R_{S}(0)=1$, and $T_p$ is the duration of coherent processing interval (CPI), which is assumed to be much larger than the symbol duration, i.e., $T_p\gg\frac{1}{B}$, where $B$ is the system bandwidth. 

In addition, the communication signals are uncorrelated to the dedicated radar signals, i.e., ${\rm{\bf E}}\left[s_k(t)s_{{\tilde k}}(t-\tau )\right]=0$, $k\in\{1,...,K\}$ and ${\tilde k}\in\{K+1,...,M_t\}$. Moreover, let $\mathbf{W}=[\mathbf{w}_{1},\mathbf{w}_{2},...,\mathbf{w}_{M_t}]\in \mathbb{C}^{M_t \times M_t}$ denote the transmit beamforming matrix. Therefore, the signal $\mathbf{x}(t)\in\mathbb{C}^{M_t\times 1}$ transmitted by BS-1 can be expressed as
\begin{equation}
\small
    \mathbf{x}(t)=\sum^{M_t}_{m=1}\mathbf{w}_ms_m(t)=\mathbf{W}\mathbf{s}(t),
\end{equation}
with the average transmit power ${\rm{\bf E}}\left[\|\mathbf{x}(t)\|^2\right]={\rm tr}(\mathbf{W}\mathbf{W}^{H})\leq P_t$, where $P_t$ is the transmit power limit.

Denote by $\mathbf{h}_k\in\mathbb{C}^{M_t\times 1}$ the communication channel vector from BS-1 to UE $k\in\mathcal{K}$. The received signal at UE $k$ is given by
\begin{align}
\small
\label{Eq_received_comm}
    y_k(t)=\mathbf{h}^{H}_k\mathbf{W}\mathbf{s}(t)+n(t),
\end{align}
where $n(t)$ represents the additive white Gaussian noise (AWGN) with variance $\sigma^2$. Since the radar waveforms $s_k(t)$, $k\in\{K+1,...,M_t\}$ can be made \emph{a-priori} known by the UE, they can be removed from the received signal in \eqref{Eq_received_comm}. The resulting signal is 
\begin{align}
\small
\label{Eq_received_comm2}
    y'_k(t)=\mathbf{h}^{H}_{k}\mathbf{w}_{k}s_k(t)+\mathbf{h}^{H}_k{\sum_{i\neq k,i=1}^{K}}\mathbf{w}_{i}s_i(t)+n(t),
\end{align}
where the second term is the interference. Therefore, the received SINR for UE $k$ is
\begin{small}
\begin{align}
\label{Comm_SINR}
    \gamma_k(\{\mathbf{w}_k\})=\frac{|\mathbf{h}^{H}_k\mathbf{w}_k|^2}{\sum_{i\neq k,i=1}^{K}|\mathbf{h}^{H}_k\mathbf{w}_i|^2+\sigma^2}.
\end{align}
\end{small}

We next derive the sensing output SNR for any given location $\mathbf{q}$ in the prescribed region $\mathcal{Q}$. Specifically, $\forall \mathbf{q}\in \mathcal{Q}$. Specifically, the sensing channel power gain of the transmit link and reflected link can be respectively expressed as
\begin{equation}
\small
    \beta_t(\mathbf{q})=\frac{\beta_0}{ \Vert \mathbf{q}- \mathbf{o} \Vert ^2},\qquad \beta_r(\mathbf{q})=\frac{\beta_0}{ \Vert \mathbf{q}-\mathbf{o}'  \Vert ^2},
\end{equation}
where $\beta_0$ denotes the channel power at the reference distance of 1 meter. Note that the received signal of the direct link from BS-1 to BS-2 can be removed based on the prior knowledge of the transmitted signal through the backhaul. Therefore, the signal received by BS-2 that is reflected by the sensing target located at $\mathbf{q}\in\mathcal{Q}$, can be written as
\begin{align}
\label{Eq_radar_receive}
    \mathbf{r}(t,\mathbf{q})=&\sqrt{\beta_t(\mathbf{q})\beta_r(\mathbf{q})}\alpha\mathbf{a}(\mathbf{q})\mathbf{b}^{H}(\mathbf{q})\mathbf{W}\mathbf{s}(t-\tau_1)+\mathbf{n}(t),
\end{align}
where $\mathbf{n}(t)\in\mathbb{C}^{M_r\times 1}$ is the AWGN with the power spectral density $N_0$, and the corresponding power is $\sigma^2=N_0B$. $\alpha$ denotes the complex-valued reflection coefficient with $|\alpha|^2$ proportional to the RCS. $\tau_1$ is the corresponding propagation delay associated with the target. $\mathbf{b}(\mathbf{q})$ and $\mathbf{a}(\mathbf{q})$ are the transmit and receive steering vectors, respectively. The signal vector $\mathbf{r}(t,\mathbf{q})$ is match-filtered by each of the transmitted waveform $s_k(t)$, and the output component corresponding to the $k$-th transmit waveform can be obtained as
\begin{align}
\small
\label{Eq_match-filter1}
    &\mathbf{y}_{k}(\mathbf{q})=\frac{1}{\sqrt{T_p}}\int_{T_p}\mathbf{r}(t,\mathbf{q})s_k^{*}(t-\tau_0)dt\nonumber\\
     &=\frac{\sqrt{\beta_t(\mathbf{q})\beta_r(\mathbf{q})}\alpha\mathbf{a}(\mathbf{q})\mathbf{b}^{H}(\mathbf{q})\mathbf{W}}{\sqrt{T_p}}\int_{T_p}\mathbf{s}(t-\tau_1)s_k^{*}(t-\tau_0)dt\nonumber\\
     &+\frac{1}{\sqrt{T_p}}\int_{T_p}\mathbf{n}(t)s_k^{*}(t-\tau_0)dt,
\end{align}
where $\tau_0$ is the delay tuned by the matched filter. The normalization factor $\frac{1}{\sqrt{T_p}}$ is applied to ensure that the noise will not be amplified. With  $T_p\gg\frac{1}{B}$, for $m, k\in\{1,...,K\}$, we have
\begin{small}
\begin{align}
\label{Eq_match_comm}
     \int_{T_p}s_m(t-\tau_1)s_k^{*}(t-\tau_0)dt\approx T_p{\rm {\bf E}}[s_m(t-\tau_1)s_k^{*}(t-\tau_0)]\nonumber\\
     =\left\{
    \begin{array}{lr}
        T_pR_{C}(\tau_0-\tau_1), &m= k,\\
        0, &m\neq k.\\
     \end{array}
\right.
\end{align}
\end{small}

For the dedicated radar signal, i.e., for $m, k\in\{K+1,...,M_t\}$, it follows from \eqref{radar_auto} that
\begin{align}
\small
\label{Eq_match_radar}
     \int_{T_p}s_m(t-\tau_1)s_k^{*}(t-\tau_0)dt=
     \left\{
    \begin{array}{lr}
        T_pR_{S}(\tau_0-\tau_1), &m=k,\\
        0, &m\neq k.\\
     \end{array}
\right.
\end{align}
Finally, for $ m\in\{1,...,K\}$ and $k\in\{K+1,...,M_t\}$, or $k\in \{1,...,K\}$ and $m\in \{K+1,...,M_t\}$, we have the output
\begin{small}
\begin{align}
\label{Eq_match_comm-radar}
     \int_{T_p}s_m(t-\tau_1)s_k^{*}(t-\tau_0)dt\approx T_p{\rm {\bf E}}[s_m(t-\tau_1)s_k^{*}(t-\tau_0)]=0.
\end{align}
\end{small}
Furthermore, for the peak of the matched filter output where $\tau_0$ aligned with the delay of the target, i.e., $\tau_0=\tau_1$, it follows from \eqref{Eq_match_comm}-\eqref{Eq_match_comm-radar} that \eqref{Eq_match-filter1} can be rewritten as
\begin{align}
\small
\label{Eq_match-output-k}
   \mathbf{y}_{k}(\mathbf{q})\approx\sqrt{T_p\beta_t(\mathbf{q})\beta_r(\mathbf{q})}\alpha\mathbf{a}(\mathbf{q})\mathbf{b}^{H}(\mathbf{q})\mathbf{w}_{k}+\mathbf{n}_k,\quad\forall k,
\end{align}
where $\mathbf{n}_k=\frac{1}{\sqrt{T_p}}\int_{T_p}\mathbf{n}(t)s^{*}_k(t-\tau_0)dt$ is the resulting noise vector which can be shown to have zero mean and covariance matrix ${\rm {\bf E}}[\mathbf{n}_k\mathbf{n}^{H}_k]=N_0\mathbf{I}$.
By concatenating the matched filter output components in \eqref{Eq_match-output-k} $\forall k=1,..., M_t$, we can obtain
\begin{align}
\small
    \mathbf{Y}(\mathbf{q})
    =\sqrt{T_p\beta_t(\mathbf{q})\beta_r(\mathbf{q})}\alpha\mathbf{a}(\mathbf{q})\mathbf{b}^{H}(\mathbf{q})\mathbf{W}+\mathbf{N},
\end{align}
where $\mathbf{N}$ is the concatenated independent and identically distributed (i.i.d.) noise matrix with zero mean and variance $N_0$.
By vectorizing $\mathbf{Y}^{H}(\mathbf{q})$, we have
\begin{equation}
    \mathbf{y}(\mathbf{q})=\sqrt{T_p\beta_t(\mathbf{q})\beta_r(\mathbf{q})}\alpha^{*}\mathbf{a}^{*}(\mathbf{q})\otimes (\mathbf{W}^{H}\mathbf{b}(\mathbf{q}))+\mathbf{\tilde{n}},
\end{equation}
where $\mathbf{\tilde{n}}\in\mathbb{C}^{M_rM_t\times 1}$ denotes the vectorized noise. 
Consequently, by applying the linear receive beamforming vector $\mathbf{v}\in\mathbb{C}^{M_rM_t\times 1}$, we have
\begin{align}
\small
\label{Eq_linear-filtering}
    z_S(\mathbf{q})=&\sqrt{T_p\beta_t(\mathbf{q})\beta_r(\mathbf{q})}\alpha^{*}\mathbf{v}^{H}\{\mathbf{a}^{*}(\mathbf{q})\otimes (\mathbf{W}^{H}\mathbf{b}(\mathbf{q}))\}+\mathbf{v}^{H}\mathbf{\tilde{n}}.
\end{align}
Note that $\mathbf{v}$ is utilized for target searching. For a potential target at location $\mathbf q\in\mathcal{Q}$, the optimal receive beamforming that maximizes the sensing output SNR is given by 
\begin{small}
\begin{equation}
\label{optimization_problem_W_r}
    \mathbf{v}=\frac{\mathbf{a}^{*}(\mathbf{q})}{\Vert\mathbf{a}(\mathbf{q})\Vert}\otimes\frac{\mathbf{W}^{H}\mathbf{b}(\mathbf{q})}{\Vert\mathbf{W}^{H}\mathbf{b}(\mathbf{q})\Vert}.
\end{equation} 
\end{small}

By using the identity  $(\mathbf{A}\otimes\mathbf{B})^{H}=\mathbf{A}^{H}\otimes\mathbf{B}^{H}$ and $(\mathbf{A}\otimes\mathbf{B})(\mathbf{C}\otimes\mathbf{D})=(\mathbf{A}\mathbf{C})\otimes(\mathbf{B}\mathbf{D})$, the maximum sensing SNR for a potential target located at $\mathbf{q} \in \mathbf{Q}$ can be expressed as
\begin{align}
    \gamma_S(\mathbf{q})
    &=\frac{T_p\beta_t(\mathbf{q})\beta_r(\mathbf{q})|\alpha|^2 |\mathbf{v}^{H}\{\mathbf{a}^{*}(\mathbf{q})\otimes\mathbf{W}^{H}\mathbf{b}(\mathbf{q})\}|^2}{\mathbf{ E}(\mathbf{v}^{H}\mathbf{\tilde{n}}\mathbf{\tilde{n}}^{H}\mathbf{v})}\nonumber\\
    &=\frac{K_{\rm CPI}\beta^2_0|\alpha|^2\|\mathbf{a}(\mathbf{q})\|^2 (\mathbf{b}^{H}(\mathbf{q})\mathbf{W}\mathbf{W}^{H} \mathbf{b}(\mathbf{q}))}{\Vert \mathbf{q}-\mathbf{o}  \Vert ^2\Vert \mathbf{q}- \mathbf{o}'\Vert^2\sigma^2},
\end{align}
where $K_{\rm CPI}=BT_{p}$ is the time-bandwidth product over one CPI, which can be interpreted as the number of symbol durations over one CPI. 

To achieve seamless sensing coverage over the prescribed region $\mathcal{Q}$, a transmit beamforming optimization problem is formulated to maximize the worst-case sensing SNR across $\mathcal{Q}$, subject to the SINR constraints for the communication UEs, which can be stated as
\begin{small}
\begin{subequations}
\label{optimization_problem1}
\begin{align}
    \max_{\mathbf{W}}\min_{\mathbf{q}\in\mathcal{Q}}&\frac{K_{\rm CPI}\beta^2_0|\alpha|^2\|\mathbf{a}(\mathbf{q})\|^2 (\mathbf{b}^{H}(\mathbf{q})\mathbf{W}\mathbf{W}^{H} \mathbf{b}(\mathbf{q}))}{\Vert \mathbf{q}-\mathbf{o}  \Vert ^2\Vert \mathbf{q}- \mathbf{o}'\Vert^2\sigma^2}\\
    \text{s.t.}\quad& \frac{|\mathbf{h}^{H}_k\mathbf{w}_k|^2}{\sum_{i\neq k,i=1}^{K}|\mathbf{h}^{H}_k\mathbf{w}_i|^2+\sigma^2}\geq \bar{\gamma}_k, k\in\mathcal{K},    \label{op1-c1}\\
    &{\rm tr}(\mathbf{W}\mathbf{W}^{H})\leq  P_{t},    \label{op1-c2}
\end{align}
\end{subequations}
\end{small}
where $\bar{\gamma}_k$ is a predefined SINR threshold for each UE $k$.
\section{Proposed Solution}
\subsection{Closed-form Solution for Single Sensing Point and One UE}

In order to gain some insights, in this subsection, we consider the special case of \eqref{optimization_problem1} with one single communication UE ($K$=1) and one single sensing point, denoted by $\mathbf{q}_0$. In this case,  \eqref{optimization_problem1} reduces to
\begin{small}
\begin{subequations}
\label{optimization_problem1-reduce}
\begin{align}
    \max_{\mathbf{W}}&\frac{K_{\rm CPI}\beta^2_0|\alpha|^2\|\mathbf{a}(\mathbf{q}_0)\|^2 (\mathbf{b}^{H}(\mathbf{q}_0)\mathbf{W}\mathbf{W}^{H} \mathbf{b}(\mathbf{q}_0))}{\Vert \mathbf{q}_0-\mathbf{o}  \Vert ^2\Vert \mathbf{q}_0- \mathbf{o}'\Vert^2\sigma^2}\\
    \text{s.t.}\quad
    &|\mathbf{h}^{H}_1\mathbf{w}_1|^2\geq \sigma^2\bar{\gamma}_1,\\
    & (\text{\ref{op1-c2}}).
\end{align}
\end{subequations}
\end{small}

By defining $\mathbf{R}_1=\mathbf{w}_1\mathbf{w}^{H}_1$ and $\mathbf{R}'=\sum_{k=2}^{M_t}\mathbf{w}_k\mathbf{w}^H_k$, and omitting those constant terms in the objective function, \eqref{optimization_problem1-reduce} can be rewritten as
\begin{subequations}
\label{optimization_problem2}
\begin{align}
\small
    \max_{\mathbf{R}_1\succeq \mathbf{0},\mathbf{R}'\succeq \mathbf{0}}  &\mathbf{b}^{H}(\mathbf{q}_0)(\mathbf{R}_1+\mathbf{R}')\mathbf{b}(\mathbf{q}_0)\\
    \text{s.t.}\quad& \mathbf{h}^{H}_1\mathbf{R}_1\mathbf{h}_1\geq \sigma^2\bar{\gamma}_1,\label{op2-c1}\\
    &{\rm tr}(\mathbf{R}_1+\mathbf{R}')\leq  P_{t},\label{op2-c1}\\
    &{\rm rank}(\mathbf{R}_1)\leq 1,\label{op2-c3}\\
    &{\rm rank}(\mathbf{R}')\leq M_t-1.\label{op2-c4}
\end{align}
\end{subequations}

Problem \eqref{optimization_problem2} is non-covex due to the rank constraints \eqref{op2-c3} and \eqref{op2-c4}. However, its optimal solution can be obtained in closed-form with the following Theorem.

\begin{theorem}
\label{Theorem_1}
An optimal solution $(\mathbf{\bar{R}}_1,\mathbf{\bar{R}}')$ to problem \eqref{optimization_problem2} is
\begin{equation}
    \mathbf{\bar{R}}_1=\mathbf{\bar{w}}_1\mathbf{\bar{w}}^{H}_1, \quad \mathbf{\bar{R}}'=\mathbf{0},
\setcounter{equation}{22}
\end{equation}
where $\mathbf{\bar{w}}_1$ is given in \eqref{optimal_beam_single} shown at the top of this page, with $\mathbf{\bar{h}}_1=\frac{\mathbf{h}_1}{\|\mathbf{h}_1\|}$, $\bar{\mathbf{b}}(\mathbf{q}_0)=\frac{\mathbf{b}(\mathbf{q}_0)}{\|\mathbf{b}(\mathbf{q}_0)\|}$, and $\bar{\mathbf{b}}^{H}(\mathbf{q}_0)_{\bot}=\frac{\mathbf{b}(\mathbf{q}_0)_{\bot}}{\|{\mathbf{b}}(\mathbf{q}_0)_{\bot}\|}$ with $\mathbf{b}(\mathbf{q}_0)_{\bot}=\mathbf{b}(\mathbf{q}_0)-(\bar{\mathbf{h}}^H_1\mathbf{b}(\mathbf{q}_0))\bar{\mathbf{h}}_1$, and $\beta_{bh}=\mathbf{\bar{h}}^{H}_1\mathbf{b}(\mathbf{q}_0)$ with $\angle\beta_{bh}\in[0,2\pi)$.
\vspace{-4mm}
\end{theorem}

\newcounter{mytempeqncnt}
\begin{figure*}[!t]
\normalsize
\setcounter{mytempeqncnt}{\value{equation}}
\setcounter{equation}{22}
\begin{align}
\label{optimal_beam_single}
\mathbf{\bar{w}}_1=
    \left\{
    \begin{array}{lr}
       \sqrt{P_t}\bar{\mathbf{b}}(\mathbf{q}_0),& 0\leq\bar{\gamma}\leq |\mathbf{h}^{H}_1\bar{\mathbf{b}}(\mathbf{q}_0)|^2\frac{P_{t}}{\sigma^2},\\
        \sqrt{\frac{\sigma^2\bar{\gamma}}{\|\mathbf{h}_1\|^2}}\bar{\mathbf{h}}_1e^{j\angle\beta_{bh}} + \sqrt{P_t-\frac{\sigma^2\bar{\gamma}}{\|\mathbf{h}_1\|^2}}\bar{\mathbf{b}}(\mathbf{q}_0)_{\bot},& |\mathbf{h}^{H}_1\bar{\mathbf{b}}(\mathbf{q}_0)|^2\frac{P_{t}}{\sigma^2}\leq\bar{\gamma}\leq\|\mathbf{h}^{H}_1\|^2\frac{P_{t}}{\sigma^2}.
     \end{array}
\right.
\end{align}
\setcounter{equation}{23}
\hrulefill
\vspace{-4mm}
\end{figure*}
\begin{proof}
Please refer to Appendix A.
\end{proof}
\vspace{-2mm}
Theorem \ref{Theorem_1} shows that for the special case of one UE and single sensing point, there is no need to transmit dedicated radar waveform. Instead, using the communication signal is sufficient for ISAC if the output SNR is the main concern for sensing. 

\subsection{General Solution for Contiguous Sensing Coverage and Multiple UEs}

In this subsection, we investigate the general problem \eqref{optimization_problem1}, in which $\mathcal{Q}$ is a contiguous sensing coverage region. To make the optimization problem \eqref{optimization_problem1} more tractable, $\mathcal{Q}$ is discretized into $L$ points $\{\mathbf{q}_l\}_{l=1}^L$. Therefore, by discarding those constant terms in the objective function,  \eqref{optimization_problem1} is reformulated as
\begin{small}
\begin{subequations}
\label{optimization_problem1-3}
\begin{align}
    \max_{\mathbf{W}}&\min_{\mathbf{q}_l,l=1,...,L}\frac{\|\mathbf{a}(\mathbf{q}_l)\|^2 (\mathbf{b}^{H}(\mathbf{q}_l)\mathbf{W}\mathbf{W}^{H}\mathbf{b}(\mathbf{q}_l))}{\Vert \mathbf{q}_l -\mathbf{o} \Vert ^2\Vert \mathbf{q}_l- \mathbf{o}'\Vert^2}\\
    & \text{s.t.}\quad(\text{\ref{op1-c1}}),\quad (\text{\ref{op1-c2}}).
\end{align}
\end{subequations}
\end{small}

\setlength{\textfloatsep}{4pt}
By introducing an auxiliary variable $\zeta$, problem \eqref{optimization_problem1-3} can be equivalently expressed as 
\begin{subequations}
\label{optimization_problem1-4}
\begin{align}
    \max_{\{\mathbf{w}_k\}_{k=1}^{M_t},\zeta}&\zeta\\
     \text{s.t.}\quad&(\text{\ref{op1-c1}}),\quad (\text{\ref{op1-c2}}), \label{op1-4-c1}\\
    &\sum_{k=1}^{M_t}|\mathbf{b}^{H}(\mathbf{q}_l)\mathbf{w}_k|^2\geq \eta_{l}\zeta,l=1,...,L,\forall k,\label{op1-4-c3}
\end{align}
\end{subequations}
where $\eta_l=\frac{\Vert \mathbf{q}_l-\mathbf{o}  \Vert ^2\Vert \mathbf{q}_l- \mathbf{o}'\Vert^2 }{\|\mathbf{a}(\mathbf{q}_l)\|^2}$. Problem \eqref{optimization_problem1-4} is still non-convex due to the non-convex constraints in \eqref{op1-c1} and  \eqref{op1-4-c3}. Fortunately, the inequality constraint \eqref{op1-c1} can be re-expressed as
\begin{small}
\begin{align}
\label{C1-conversion}
  \sum_{i=1}^{K}|\mathbf{h}^{H}_k\mathbf{w}_i|^2+\sigma^2\leq \left(1+\frac{1}{\bar{\gamma}_k}\right)|\mathbf{h}^{H}_k\mathbf{w}_k|^2,\quad k\in\mathcal{K}.
\end{align}
\end{small}

\begin{algorithm}[t]
\caption{SCA for solving \eqref{optimization_problem1-4}}
\label{alg:SCA}
\begin{algorithmic}[1]
\REQUIRE
$P_t$, $\bar{\gamma}_k$, $\mathbf{h}_k$, $\sigma$, $\mathbf{b}(\mathbf{q}_l)$.
\ENSURE
$\{\mathbf{w}^{\mathrm{opt}}_{k}\}$.
\STATE{Initialize $\mathbf{w}_k^{(0)},\forall k$. Let $\kappa$=0.}
\REPEAT
\STATE {Solve the convex problem \eqref{optimization_problem1-5} for the given local point $\{\mathbf{w}^{(\kappa)}_k\}$, and denote the optimal solution as $\{\mathbf{w}^{\mathrm{opt}}_k\}$.}
\STATE {Update the local point $\{\mathbf{w}^{(\kappa+1)}_k\}=\{\mathbf{w}^{\mathrm{opt}}_k\}$.}
\STATE{Update $\kappa=\kappa+1$.}
\UNTIL{The increase of the objective value of \eqref{optimization_problem1-5} is below a predefined threshold $\epsilon$.}
\end{algorithmic}
\end{algorithm}
\setlength{\textfloatsep}{1pt}
Without loss of optimality, by rotating the phase of $\mathbf{w}_k$ accordingly, we may assume that $\mathbf{h}^{H}_k\mathbf{w}_k$ is real and positive. Thus, \eqref{C1-conversion} can be equivalently written as the following second-order cone (SOC) constraints:
\begin{small}
\begin{align}
       & \|
        \mathbf{h}^{H}_k\mathbf{w}_1 ,...,\mathbf{h}^{H}_k\mathbf{w}_K,\sigma\|_2
    \leq\sqrt{1+\frac{1}{\bar{\gamma}_k}}\mathbf{h}^{H}_k\mathbf{w}_k, \nonumber\\
    &{\rm Im}\{\mathbf{h}^{H}_k\mathbf{w}_k\}=0, \quad k\in\mathcal{K}.
\end{align}
\end{small}
\setlength{\textfloatsep}{4pt}
Furthermore, to deal with the non-convex constraint \eqref{op1-4-c3}, we may apply the SCA technique \cite{Yzeng2019}. Specifically, by using the fact that the first order Taylor series of the differentiable convex functions serve as global lower bound, we have 
\begin{align}
\label{Taylor-expansion}
\small
       &\sum_{k=1}^{M_t}|\mathbf{b}^{H}(\mathbf{q}_l)\mathbf{w}_k|^2=\sum_{k=1}^{M_t}\mathbf{w}^{H}_k\mathbf{b}(\mathbf{q}_l)\mathbf{b}^{H}(\mathbf{q}_l)\mathbf{w}_k\triangleq g(\mathbf{w}_k)\nonumber\\ &\geq g(\mathbf{w}^{(\kappa)}_k)+2\sum_{k=1}^{M_t}{\rm Re}\{\mathbf{w}^{(\kappa)H}_k\mathbf{b}(\mathbf{q}_l)\mathbf{b}^{H}(\mathbf{q}_l)(\mathbf{w}_k-\mathbf{w}^{(\kappa)}_k)\}\nonumber\\  
       &=g(\mathbf{w}_k,\mathbf{w}^{(\kappa)}_k), \quad \forall \mathbf{w}_k,
\end{align}
where $\mathbf{w}^{(\kappa)}_k$ is the local point obtained at the $\kappa$-th iteration.
As such, the non-convex optimization problem \eqref{optimization_problem1-4} can be recast as the following convex problem
\begin{small}
\begin{subequations}
\label{optimization_problem1-5}
\begin{align}
    \max_{\{\mathbf{w}_k\}_{k=1}^{M_t},\zeta}&\zeta\\
    \text{s.t.}\quad&               \|
        \mathbf{h}^{H}_k\mathbf{w}_1 ,...,\mathbf{h}^{H}_k\mathbf{w}_K,\sigma\|_2
    \leq\sqrt{1+\frac{1}{\bar{\gamma}_k}}\mathbf{h}^{H}_k\mathbf{w}_k, k\in\mathcal{K},\\
    & \sum_{m=1}^{M_t}\|\mathbf{w}_m\|^2\leq  P_{t},\\
    & \eta_{l}\zeta-g(\mathbf{w}_k,\mathbf{w}^{(\kappa)}_k)\leq 0,l=1,...,L,\forall k,\\
    &{\rm Im}(\mathbf{h}^{H}_k\mathbf{w}_k)=0,k\in\mathcal{K}.
\end{align}
\end{subequations}
\end{small}
\setlength{\textfloatsep}{5pt}
The optimal solution of the convex optimization problem \eqref{optimization_problem1-5} can be efficiently obtained based on the standard convex optimization techniques \cite{CChi2017}, or readily available
software toolbox, such as CVX \cite{MGrant2018}. Furthermore, due to the global lower bound in \eqref{Taylor-expansion},  it can be shown that the corresponding objective value gives at least a lower bound to that of problem \eqref{optimization_problem1-4}. Thus, by iteratively solving the convex optimization problem \eqref{optimization_problem1-5} with the updated local point $\{\mathbf{w}^{(\kappa+1)}_k\}=\{\mathbf{w}^{\mathrm{opt}}_k\}$, an efficient local optimal solution of \eqref{optimization_problem1-4} and hence that of \eqref{optimization_problem1-3} can be obtained, which is summarized in Algorithm \ref{alg:SCA}.

\section{Simulation Results}
In this section, simulation results are provided to evaluate the performance of our proposed design. The height of both BSs is $H=10$ meters, and they are equipped with uniform planar array (UPA) on the $xz$-plane, consisting of $M_t=M_r=64$ antenna elements with the inter-element spacing $d=\frac{\lambda}{2}$, where $\lambda$ is the carrier wavelength. Therefore, the transmit steering vector towards location $\mathbf{q}$ is given by
\begin{align}
    \mathbf{b}(\mathbf{q}) =&[1,...,e^{j\frac{2\pi}{\lambda}d[(M_x-1)\sin\theta(\mathbf{q})\cos\phi(\mathbf{q})}]^{T}\nonumber\\
    &\otimes [1,...,e^{j\frac{2\pi}{\lambda}d(M_z-1)\cos\theta(\mathbf{q})]}]^{T},
\end{align}
where $\theta(\mathbf{q})$ and $\phi(\mathbf{q})$ denote elevation and azimuth angles, respectively, and $M_x=8$ and $M_z=8$ denote the number of antenna elements along $x$-axis and $z$-axis, respectively.  
The bandwidth is $B = 100$ MHz, and the CPI duration is $T_p=1$ ms. The transmit power is $P_t = 20$ dBm, and the reference channel power is $\beta_0=-40$ dB. The noise power is $\sigma^2=-90$ dBm. The channel coefficients for the communication UEs are generated based on Rician fading with the Rician factor $G$=10.
\setlength{\textfloatsep}{4pt}
\vspace{-5mm}
\begin{figure}[h]
\center
\includegraphics[width=2.8in]{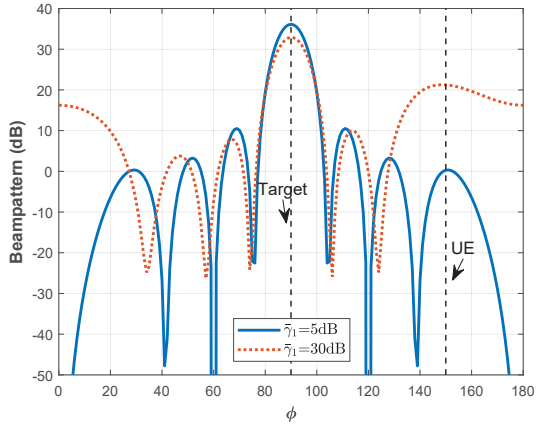}\\
\caption{Beampatterns for one UE and single sensing point. }\label{Fig_singlecase}
\vspace{-4mm}
\end{figure}

For the special case of one UE and single sensing point, Fig. \ref{Fig_singlecase} shows the beampatterns along $\phi$ with the optimal beamforming in \eqref{optimal_beam_single}, where the sensing point is located at $(\theta,\phi)=(90^{\circ}, 90^{\circ})$ with 50 meters to BS-1, while UE is located at $(\theta,\phi)=(135^{\circ}, 150^{\circ})$ with 30 meters to BS-1. The beampattern is computed by
    $\frac{\|\mathbf{b}^{H}(\mathbf{q})\mathbf{W}\|^2}{{\rm tr}(\mathbf{W}\mathbf{W}^{H})}$.
It can be observed that as the communication SNR threshold $\bar{\gamma}_1$ increases from 5 dB to 30 dB, the optimal beamforming vector $\bar{\mathbf{w}}_1$ will result in a larger gain towards the UE, while a slightly decreased gain towards the sensing target. This is expected due to the trade-off between sensing and communication when $\bar{\gamma}_1$ is large, as can be inferred from \eqref{optimal_beam_single}.

Next, we investigate the performance of our proposed design for the general case with regional sensinng coverage and multiple UEs. The top view on the geometry of the coverage region $\mathcal{Q}$ is shown in Fig. \ref{Fig_vertical}. Specifically, $\mathcal{Q}$ is a rectangular region with size of 50 {\rm meters} $\times$ 50 {\rm meters}, and is 10 meters above the ground. Therefore, $\mathcal{Q}$ lies in an angle space with $\theta=90^{\circ}$ and $\phi \in [45^{\circ}, 135^{\circ}]$. For ease of optimization, $\mathcal{Q}$ is uniformly discretized into $L = 2500$ positions. Additionally, the number of communication UEs is $K = 2$, which are located at $(\theta,\phi)=(135^{\circ}, 30^{\circ})$ and $(\theta,\phi)=(135^{\circ}, 150^{\circ})$, respectively. Furthermore, the communication SINR threshold $\bar{\gamma}_k$ for both UEs are set to 20 dB.
\begin{figure}[h]
\center
\includegraphics[width=2.8in]{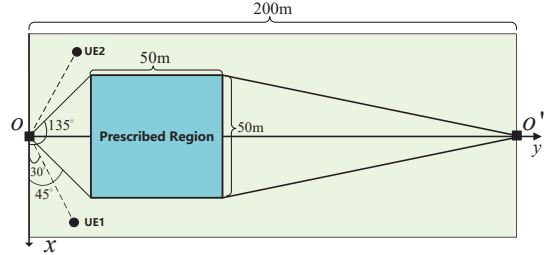}\\
\caption{Top view on the geometry of the prescribed region.}\label{Fig_vertical}
\vspace{-2mm}
\end{figure}

Based on the above settings, Fig. \ref{Fig_beam_Comparison} shows the beampattern of our proposed beamforming design. As a benchmark comparison, we also consider the communication-only beamforming design, which minimizes the transmit power subject to the SINR requirement for each UE, without considering the sensing coverage requirement  \cite{EBjornson2014}. It is observed that in the prescribed coverage region, the beampattern gain of our proposed beamforming design is higher than that of its communication-only counterpart. Furthermore, our proposed design can form a relatively flat beampattern for seamless coverage. By contrast, the communication-only beamforming design does not consider the sensing coverage requirement, leading to poor beamforming gains in the sensing coverage region. 
\vspace{-4mm}
\begin{figure}[h]
\center
\includegraphics[width=2.9in]{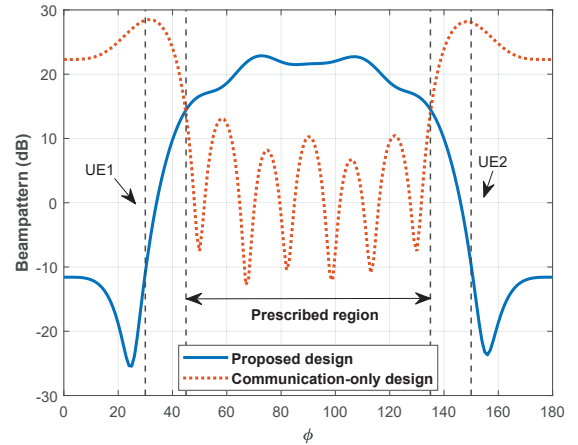}\\
\caption{Beampatterns of the proposed and communication-only beamforming designs for regional sensing coverage and multiple UEs.}\label{Fig_beam_Comparison}
\vspace{-2mm}
\end{figure}

\begin{figure}
 \centering
\begin{tabular}{ c @{\hspace{1pt}} c }
\includegraphics[width=2.5in]{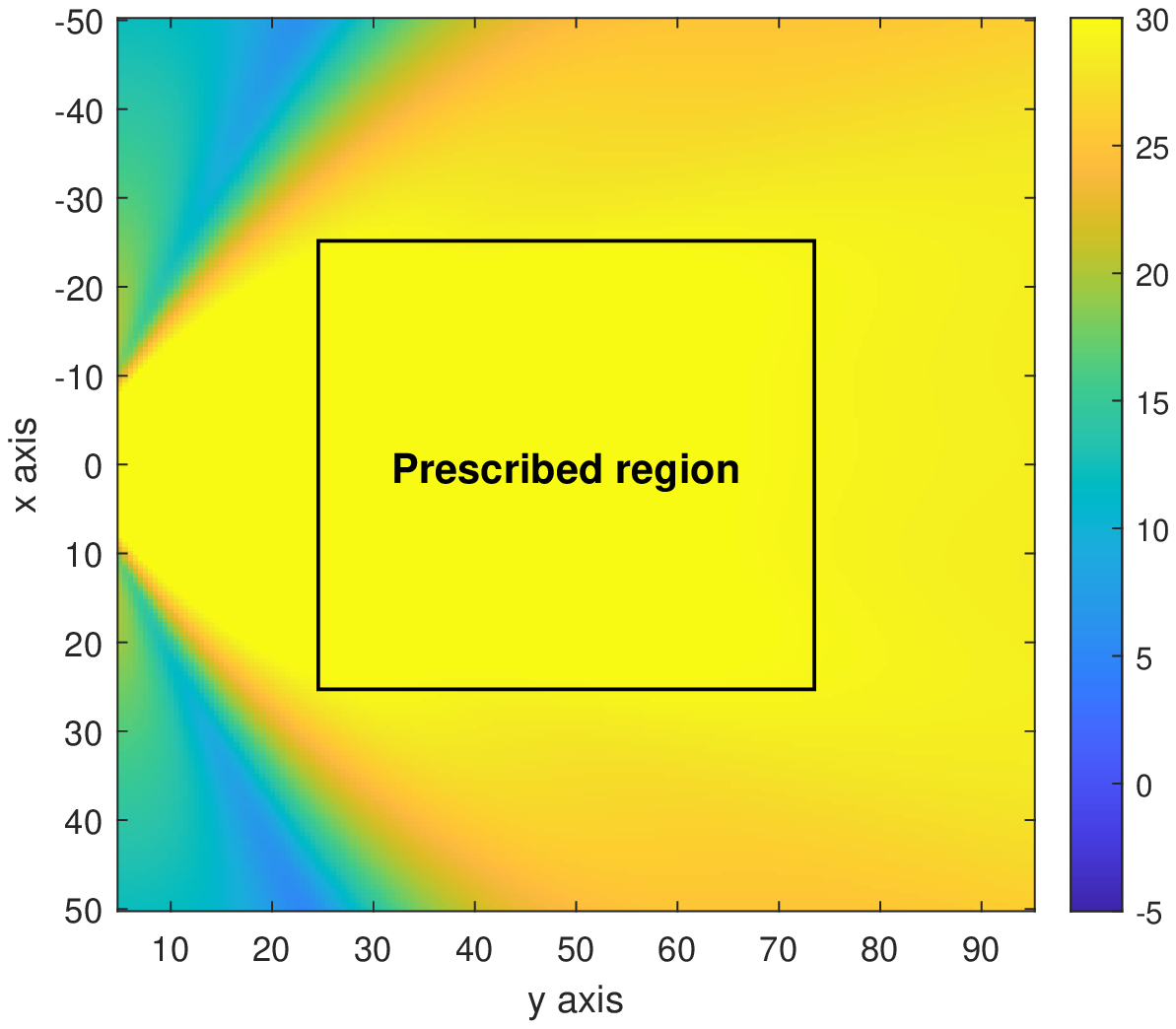}
\\
\small (a) Proposed beamforming design\\
\includegraphics[width=2.5in]{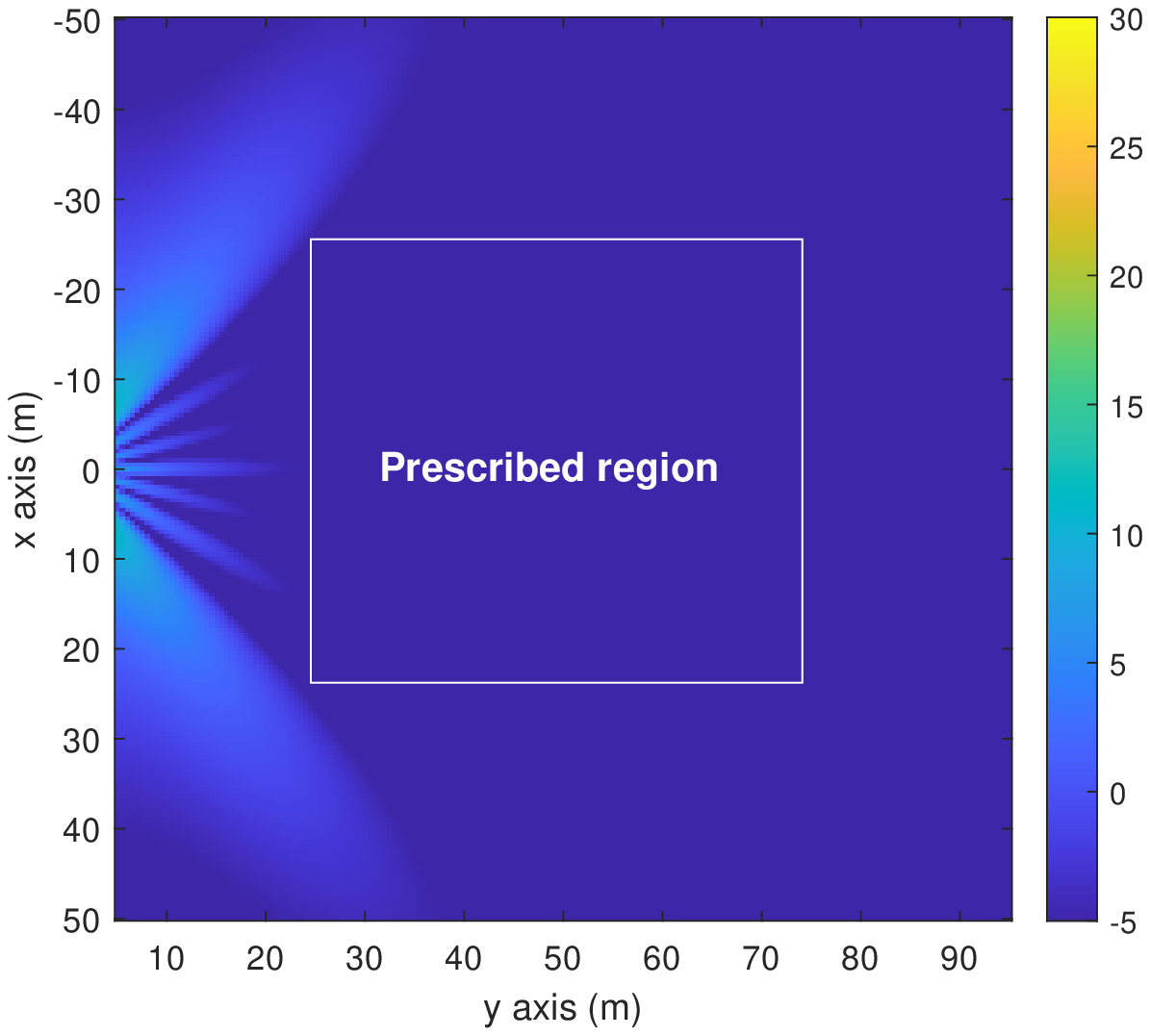}\\
 \small (b) Communication-only beamforming
\end{tabular}
\caption{Sensing SNR distribution of two beamforming designs.}
\label{SNR_comaprison}
\vspace{-2mm}
\end{figure}

Fig. \ref{SNR_comaprison} further compares our proposed beamforming design with the communication-only counterpart, in terms of the sensing SNR distribution over the prescribed coverage region. It is observed that the SNR obtained by our proposed design is much higher than that of communication-only beamforming design. Therefore, our proposed ISAC design is able to achieve seamless sensing coverage over the prescribed region, while satisfying the communication requirement of UEs. 

\section{Conclusion}
In this paper, we have investigated beamforming design for the coverage optimization in a cellular ISAC system. Our proposed beamforming design aimed to maximize the worst-case SNR in a prescribed region, while guaranteeing the SINR requirement for the UEs. A closed-form expression of optimal beamforming vector was derived for the special case with one sensing point and single UE. Besides, an SCA based algorithm was proposed to obtain an efficient local optimal solution to the general case with  multiple communication UEs and contiguous 
prescribed sensing coverage region. Simulation results demonstrated the performance gain of our proposed beamforming design over the benchmark scheme.
\vspace{-2mm}
\section{Appendix}
\subsection{Proof of Theorem 1}
By omitting the rank constraints \eqref{op2-c3} and \eqref{op2-c4}, we can obtain a rank-relaxed problem of \eqref{optimization_problem2}.

\begin{lemma}
Let $(\mathbf{\hat{R}}_1,\hat{\mathbf{R}}')$ be an optimal solution to the rank-relaxed problem of \eqref{optimization_problem2}, then the constructed new solution
\begin{align}
\small
    (\mathbf{R}^{\star}_1,\mathbf{R}'^{\star})=(\hat{\mathbf{R}}_1+\hat{\mathbf{R}}',\mathbf{0}),\nonumber
\end{align}
is also its optimal solution. 
\end{lemma}
\begin{proof}
To show Lemma 1, we only need to show that $(\mathbf{R}^{\star}_1,\mathbf{R}'^{\star})$ satisfy all constraints of the rank-relaxed problem of \eqref{optimization_problem2} and achieves the same objective value as $(\mathbf{\hat{R}}_1,\hat{\mathbf{R}}')$. First, by substituting the solution $ (\mathbf{R}^{\star}_1,\mathbf{R}'^{\star})$ into the rank-relaxed problem, we can observe that the objective value is the same as that of $(\mathbf{\hat{R}}_1,\hat{\mathbf{R}}')$, and \eqref{op2-c1} is satisfied as well.

Next, since both $\hat{\mathbf{R}}_1$ and $\hat{\mathbf{R}}'$ are positive semidefinite, for the new solution $(\mathbf{R}^{\star}_1,\mathbf{R}'^{\star})$, we have
\begin{align}
\small
    \mathbf{h}^{H}_1\mathbf{R}^{\star}_1\mathbf{h}_1= \mathbf{h}^{H}_1(\hat{\mathbf{R}}_1+\hat{\mathbf{R}}')\mathbf{h}_1\geq \mathbf{h}^{H}_1\hat{\mathbf{R}}_1\mathbf{h}_1\geq\sigma^2\bar{\gamma}_1.
\end{align}
Thus, the constraint \eqref{op2-c1} still holds with the solution $(\mathbf{R}^{\star}_1,\mathbf{R}'^{\star})$. Therefore, $(\mathbf{R}^{\star}_1,\mathbf{R}'^{\star})$ is also an optimal solution to the rank-relaxed problem. This completes the proof.
\end{proof}
\vspace{-2mm}
Based on Lemma 1, without loss of optimality, we may let $\mathbf{R}'=0$ so that the rank-relaxed problem of \eqref{optimization_problem2} reduces to
\begin{align}
\small
\label{optimization_problem2-2}
    \max_{\mathbf{R}_1\succeq \mathbf{0}}&  \mathbf{b}^{H}(\mathbf{q}_0)\mathbf{R}_1\mathbf{b}(\mathbf{q}_0)\\
    \text{s.t.}\quad& \mathbf{h}^{H}_1\mathbf{R}_1\mathbf{h}_1\geq \sigma^2\bar{\gamma}_1,\nonumber\\
    &{\rm tr}(\mathbf{R}_1)\leq  P_{t}.\nonumber
\end{align}
\eqref{optimization_problem2-2} resembles the problem for simultaneous wireless information and power transfer (SWIPT) in \cite{RZhang2013TWC}, for which the optimal solution is obtained as $\bar{\mathbf{R}}_1=\bar{\mathbf{w}}_1\bar{\mathbf{w}}^{H}_1$, with $\bar{\mathbf{w}}_1$ given by \eqref{optimal_beam_single}. 

As a result, an optimal solution to the rank-relaxed problem of \eqref{optimization_problem2} is $(\bar{\mathbf{R}}_1,\mathbf{0})$. It is obvious that this solution also satisfies the rank constraints \eqref{op2-c3} and \eqref{op2-c4} of \eqref{optimization_problem2}. Therefore, it must be an optimal solution to \eqref{optimization_problem2} as well. This completes the proof of Theorem 1.

\end{document}